\documentclass[11pt]{article}
\usepackage{amssymb} 
\usepackage{amsmath}     
\usepackage{amsthm}
\usepackage{todonotes}

\usepackage{a4wide}
\usepackage{graphicx}
\usepackage{hyperref}

\newtheorem{thm}{Theorem}
\newtheorem{lem}{Lemma}

\newtheorem{cor}{Corollary}

\newenvironment{keyword}{\par{\noindent\bf Keywords:}}

\begin{document}

\title{Using the WOWA operator in robust discrete optimization problems}

\author{Adam Kasperski\thanks{Corresponding author}\\
   {\small \textit{Department of Operations Research,}}
  {\small \textit{Wroc{\l}aw University of Technology,}}\\
  {\small \textit{Wybrze{\.z}e Wyspia{\'n}skiego 27,}}
  {\small \textit{50-370 Wroc{\l}aw, Poland,}}
  {\small \textit{adam.kasperski@pwr.edu.pl}}
  \and
  Pawe{\l} Zieli{\'n}ski\\
    {\small \textit{Department of Computer Science (W11/K2), }}
  {\small \textit{Wroc{\l}aw University of Technology,}}\\
  {\small \textit{Wybrze{\.z}e Wyspia{\'n}skiego 27,}}
  {\small \textit{50-370 Wroc{\l}aw, Poland,}}
  {\small \textit{pawel.zielinski@pwr.edu.pl}}} 
  
  \date{}
    
\maketitle

\begin{abstract}
In this paper a class of discrete optimization problems with uncertain costs is discussed. The uncertainty is modeled by introducing a  scenario set containing a finite number of cost scenarios. A probability distribution 
over the set of scenarios
is available. In order to choose a solution the weighted OWA criterion (WOWA) is applied. This criterion allows decision makers to take into account both probabilities for scenarios and the degree of pessimism/ optimism. In this paper the complexity of the considered 
class of discrete optimization problems
 is described  and some exact and approximation algorithms for solving it are proposed. 
 Applications to a selection  and the assignment problems, together with results of computational tests are shown.
\end{abstract}

\section{Introduction}

Most practical decision making problems arise in a risky or uncertain environment, which means that an outcome of each decision is unknown and depends on a state of the world, which may occur with some positive probability. 
If  probabilities for the  states of the world are available, then each decision leads to a \emph{lottery}, i.e. a probability distribution over the set of all possible outcomes. A decision problem can  then be reduced to establishing 
an ordering  of the set of lotteries. According to the classic expected utility theory by von Neumann and Morgenstern~\cite{VNM53, LR57}, 
the decision maker can assign an utility to each outcome, if
he  accepts some simple and appealing axioms.
He can then compute an expected utility of each lottery and choose a decision which leads to a lottery with the largest expected utility. 

The expected utility can be seen as a weighted average of  outcomes, where the weight of each outcome is just the probability of obtaining it. Thus, in the von Neumann and Morgenstern theory, the weights are independent of the outcomes and other probabilities of the lottery. However, it has been observed in human behavior that this assumption is often violated (see~\cite{DW01} for a deeper discussion on this topic). Many decision makers pay more attention to unfavorable outcomes and would assign larger weights to such outcomes. In such a situation the weight of each outcome  depends not only on its probability, but also on its rank in the lottery. Such weights may better reflect the pessimism/optimism of decision makers. A theory of such rank dependent, transformed probabilities was introduced by  Quiggin~\cite{Q82} (see also~\cite{ST00}).

In many practical situations  the probabilities of scenarios are not available. We then obtain a decision problem under uncertainty. In this case, decision makers may assign subjective probabilities to scenarios~\cite{SA54} and compute the expected utility with respect to these subjective probabilities. However, determining the subjective probabilities may be not an easy task. An alternative approach is to apply some decision criteria such as the min-max, min-max regret, Hurwicz, or Laplace ones. In particular, in the Laplace criterion we apply the principle of insufficient reason and assign equal probability to each scenario. Each decision is then evaluated as the average utility of all possible outcomes. For a deeper discussion on decision making under uncertainty and description of the criteria we refer the reader to~\cite{LR57}.

In this paper we discuss a class of discrete optimization problems, in which a finite set of feasible solutions is specified.  In the deterministic case a cost of each solution is known and a decision problem consists in choosing a  solution
 with the minimum cost. Discrete optimization problems are often represented as integer programming ones, in which the set of feasible solutions is described in compact form by a system of constraints. A class of deterministic discrete optimization problems was described, for example, in~\cite{PS82}. In many practical situations, the cost of each solution is unknown and depends on a state of the world which may occur with some positive probability.
Each state of the world induces a cost scenario.
 A \emph{scenario set} containing all possible cost scenarios
 is  part of the input. In this paper we assume that this scenario set contains a finite number of explicitly listed scenarios. We also assume that  probabilities for the scenarios are available. Notice, that under uncertainty, the principle of insufficient reason can be applied, which assigns equal probabilities to scenarios~\cite{LR57}.  In order to choose a solution, we will apply the \emph{Weighted Ordered  Weighted Averaging} (WOWA for short) operator, proposed by Torra~\cite{TT11}. Given a solution, this operator allows us to define a rank-dependent weight for this solution under each scenario. This weight can be seen as a distorted scenario probability and the WOWA criterion is then 
 a special case of the Choquet integral with respect to distorted probabilities~\cite{GR11}.  We can evaluate each solution as a weighted average of its costs over all scenarios. The WOWA criterion contains basic criteria used in decision making under risk and uncertainty, such as the expectation (weighted mean), maximum, minimum, Hurwicz, and Laplace ones. Furthermore, if the principle of insufficient reason is applied, then WOWA becomes the OWA criterion proposed by Yager~\cite{YA88}. 

If the uncertainty is represented by a  discrete  uncertainty set, it is common to use 
the robust approach~\cite{KY97} to  compute a solution.
In this approach we assume that decision makers are risk averse and we seek a solution which minimize the cost in 
the worst case. This leads to applying the min-max or min-max regret criteria to choose a solution. 
The traditional robust approach has, however, several drawbacks. The min-max criterion is extremely conservative and it is not difficult to show examples in which it gives unreasonable solutions~\cite{LR57}. In particular, applying this criterion we may get a solution which is not Pareto optimal. Furthermore, the so-called \emph{drowning effect} may also appear~\cite{DF99}. If the costs under some scenario are large in comparison with the costs under the remaining scenarios, then only this bad scenario is taken into account in the process of computing a solution (information connected with the remaining scenarios is ignored). Hence, in many applications a criterion which takes into account all (or at least a subset) of scenarios  is required.
The traditional robust approach assumes also
that no probabilities are available for the scenarios,
which is not always true. By using the WOWA criterion we can overcome this drawback. We can use the information connected with scenario probabilities and soften the very conservative min-max criterion. Furthermore, the WOWA criterion is consistent with the theory of rank-dependent probabilities and, in consequence, can better reflect the real attitude of decision makers towards  risk. This is particularly important when decisions are not repetitious, i.e. they are implemented only once. The WOWA operator allows us to establish a link between the stochastic and robust optimization frameworks. The distorted (rank-dependent) probabilities allows us to establish a trade-off between the expected and the maximum solution costs.

In this paper we focus on the computational properties of the considered problem. Since the maximum criterion is a special case of the WOWA criterion, all  negative results known for the robust min-max problems remain valid if the WOWA criterion is used. Unfortunately, the min-max versions of all basic discrete optimization problems become NP-hard even for two scenarios. This is the case for the shortest path, minimum spanning tree, minimum assignment, minimum cut, or minimum selecting items problems~\cite{KY97, ABV08, AV01}. All these aforementioned problems become strongly NP-hard and also hard to approximate when the number of scenarios is  part of the input~\cite{KZ09, KZ11, KZ13}. Furthermore, when the OWA operator is used to choose a solution, then 
network problems (the shortest path, minimum spanning tree, minimum assignment, minimum cut)
are not at all approximable~\cite{KZ15}. However, for an important case of nondecreasing weights in the OWA operator, there exists an approximation algorithm with some guaranteed worst case ratio and the aim of this paper is to generalize this algorithm to  the more general WOWA criterion. In the existing literature, the OWA operator and the more general Choquet integral have been recently applied to some multiobjective optimization problems in~\cite{GS12, GPS10, FPP14}. In these papers the authors propose some exact methods for solving the problems, which are based on a MIP formulation and a branch and bound method.

This paper is organized as follows. In Section~\ref{secForm}, we present the problem formulation and show a motivation for using WOWA as a criterion for choosing a solution under risk and uncertainty. In Section~\ref{secComp}, we recall some known complexity results for the considered problem. In Section~\ref{secAppr}, we propose an approximation algorithm for solving the problem, which can be applied to a large class of discrete optimization problems. Section~\ref{secMIP} describes a method of constructing a mixed integer programming formulation, which can be used to solve the considered problem exactly. This method will be adopted from~\cite{OS09}. Finally, in Section~\ref{secExp}, we show  applications of the proposed model to a selection and the assignment problems. This section also contains results of computational tests, which describe the efficiency of the MIP formulation and the quality of the solutions that are 
returned by the approximation algorithm designed in Section~\ref{secAppr}.

\section{Problem formulation}
\label{secForm}
Let $E=\{e_1,\dots,e_n\}$ be a finite set of elements and let $\Phi\subseteq 2^E$ be a set of feasible solutions. In a deterministic case, each element $e_i\in E$ has a nonnegative cost $c_i$ and we seek a feasible solution $X\in \Phi$, which minimizes the total cost $F(X)=\sum_{e_i\in X} c_i$. We denote such a deterministic discrete optimization problem by $\mathcal{P}$. This formulation encompasses a wide class of problems (see, e.g.,~\cite{PS82, AH93}). We obtain, for example, a class of network problems by identifying $E$ with edges of a graph $G$ and $\Phi$ with some objects in $G$ such as paths, spanning trees,  matchings, or cuts. Usually, $\mathcal{P}$ is represented as an integer 0-1~programming problem whose constraints describe $\Phi$ in compact form.

Assume that the element costs are uncertain and their values depend on a state of the world which may occur with 
some positive probability. Each such a state of the world induces an element cost scenario (scenario for short) $\pmb{c}_j=(c_{j1},\dots,c_{jn})$. Let scenario set $\Gamma=\{\pmb{c}_1,\dots,\pmb{c}_K\}$ contain $K$ explicitly listed scenarios. Let $\pmb{p}=(p_1,\dots,p_K)$ be
a vector of scenario probabilities, i.e. $p_j$ is the probability of the event that scenario $\pmb{c}_j$ will occur. The cost of  a solution $X$ depends on scenario $\pmb{c}_j\in \Gamma$ and we will denote it by $F(X,\pmb{c}_j)=\sum_{e_i\in X} c_{ji}$.  Choosing a solution $X$ leads to a \emph{lottery}, i.e. a probability distribution $(p_1 F(X,\pmb{c}_1),\dots,p_K F(X,\pmb{c}_K)$) over the costs of $X$ under scenarios in $\Gamma$. In order to choose the best solution, we need to evaluate each lottery. To do this, we should assign a weight to each scenario and compute the weighted average solution cost.  Under the assumption that the weight of the $j$th scenario  is equal to~$p_j$, we obtain that the weighted average is just the expected solution cost, i.e. $f(X)=\mathbf{E}[X]=\sum_{j\in [K]} p_jF(X,\pmb{c}_j)$.
\begin{figure}[ht]
\centering
\includegraphics[scale=.9]{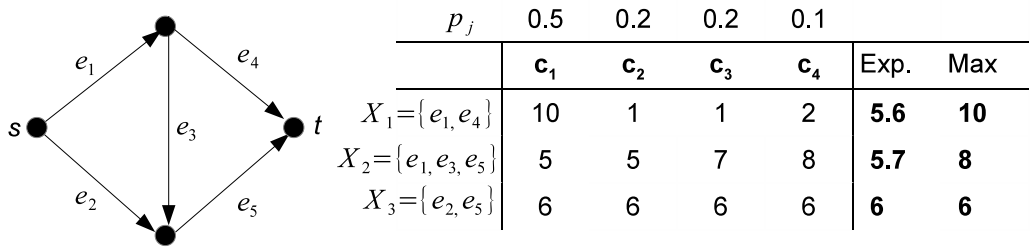}
\caption{A sample \textsc{Shortest Path} problem with four scenarios $\pmb{c}_1=(5,6,0,5,0)$, $\pmb{c}_2=(1,6,4,0,0)$, $\pmb{c}_3=(1,6,6,0,0)$, and $\pmb{c}_4=(2,6,6,0,0)$. The costs of all three paths under all scenarios are shown in the table.}
\label{figex1}      
\end{figure}

Consider the sample \textsc{Shortest path} problem depicted in Figure~\ref{figex1}. The set of elements $E=\{e_1,\dots,e_5\}$, contains~5 arcs of network $G$ and the set of feasible solutions consists of three paths $X_1$, $X_2$,
 and $X_3$ from $s$ to $t$ in $G$. There are 4 costs scenarios with the probabilities $0.5$, $0.2$, $0.2$ and $0.1$, respectively, and the costs of each path under these scenarios are shown in Figure~\ref{figex1}. Path $X_1=\{e_1,e_4\}$ has the smallest expected cost and thus should be chosen when the expected value is used. However, this choice may be unreasonable for some risk averse or pessimistic decision makers. Observe that the probability that the path $X_1$ will have a large cost equal to~10 is equal to~0.5
  which may be too large and cause some decision makers to reject $X_1$. On the other hand, the path $X_3=\{e_2,e_5\}$ has the smallest maximum cost and should be chosen when the min-max criterion is used and the probabilities of scenarios are ignored.  Notice that the path $X_3$ has a deterministic cost equal to~6. 
  However, some decision makers may feel  that path $X_2=\{e_1,e_3,e_5\}$ is better, since the probability that the cost of $X_2$ will be less than~6 is equal to~0.7 and the probability that $X_2$ will have a large cost, equal to~8, is only 0.1. The sample problem illustrates  that there is a need of defining aggregation weights, which would depend not only on the scenario probabilities, but also on the rank positions of the costs of a solution under scenarios. For example, risk averse decision makers would assign a weight larger that 0.5 to scenario $\pmb{c}_1$, when solution $X_1$ is considered.

Before we discuss such a criterion, which fulfills the above requirements, we recall
an aggregation criterion, called the \emph{Ordered Weighted Averaging}  operator (OWA for short)
proposed by Yager in~\cite{YA88}. 
Let $\pmb{w}=(w_1,\dots,w_K)$  be a weight vector such that $w_j\in [0,1]$ for each $j\in [K]$, $\sum_{j\in [K]} w_j=1$
(we use $[K]$ to denote the set $\{1,\dots,K\}$). 
Given a vector of reals $\pmb{a}=(a_1,\dots,a_K)$, let $\sigma$ be a sequence of $[K]$ such that $a_{\sigma(1)}\geq \dots \geq a_{\sigma(K)}$. Then
$$
{\rm owa}_{(\pmb{v})}(\pmb{a})=\sum_{j\in [K]} w_j a_{\sigma(j)}.
$$
The choice of particular weight vectors~$\pmb{w}$
leads to well known criteria in decision making under uncertainty  (see, e.g.,~\cite{KZ15}). Indeed,
if $w_1=1$ and $w_j=0$ for $j=2,\dots,K$, then the OWA criterion becomes the maximum.
If $w_K=1$ and $w_j=0$ for $j=1,\dots,K-1$, then it reduces to the minimum.
More generally, if $w_k=1$ and $w_j=0$ for $j\in [K]\setminus\{k\}$, then the  OWA criterion is the $k$-th largest cost and, in particular, when $k=\lfloor K/2 \rfloor +1$, then the $k$-th largest cost is the median.
If $w_j=1/K$ for all $j\in [K]$, then OWA is the average. Finally, if $w_1=\alpha$ and $w_K=1-\alpha$, for some fixed $\alpha\in [0,1]$, and $w_j=0$ for the remaining weights, then the Hurwicz pessimism-optimism criterion is obtained. 

Using  OWA  it is not easy to take the probabilities of scenarios into account. In particular, the expected value is not a special case of OWA. In the following, we will present an aggregation criterion, called the \emph{Weighted Ordered Weighted Averaging}
operator (WOWA for short)
proposed by Torra~\cite{TT11}.  This criterion generalizes OWA and allows us to define  rank-dependent weights which also depend on scenario probabilities.

Let $\pmb{v}=(v_1,\dots,v_K)$  be a weight vector  such that $v_j\in [0,1]$ for each $j\in [K]$ and $\sum_{i\in [K]} v_j=1$.
Let $w^*$ be  a continuous nondecreasing function on $[0,1]$, $w^*:[0,1]\rightarrow[0,1]$. The 
domain interval $[0,1]$ is partitioned by points $0=\frac{0}{K}<\frac{1}{K}<\frac{2}{K}<\cdots<\frac{K}{K}=1$.
The function~$w^*$ is linear on each subinterval $\left[\frac{j-1}{K},\frac{j}{K} \right]$, $j\in [K]$, 
($w^*$ is piecewise linear function)
and 
satisfies, for a given weight vector~$\pmb{v}$, 
the following equations: $w^*(0)=0$ and $w^*\left(\frac{j}{K}\right)= \sum_{i\leq j} v_i$  for $j\in [K]$. Observe that $w^*$ is uniquely defined by $\pmb{v}$.
In this paper we also make the assumption: $v_1\geq v_2\geq \dots \geq v_K$. Thus
$w^*$ is additionally  a concave function. Figure~\ref{fig1} presents three sample functions $w^*$ for $K=5$
for three weight vectors.
\begin{figure}[ht]
\centering
\includegraphics[scale=.55]{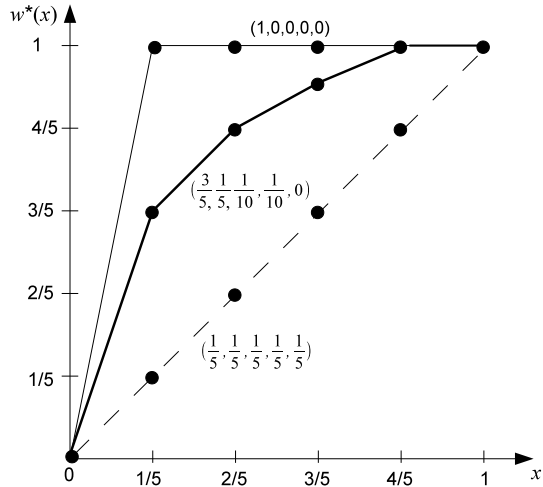}
\caption{Three sample functions $w^*$ for $K=5$}
\label{fig1}      
\end{figure}

 Let $\pmb{p}=(p_1,\dots,p_K)$ be an additional weight vector such that $p_j\in [0,1]$ for each $j\in [K]$, $\sum_{j\in [K]} p_j=1$.
 Given a vector of reals $\pmb{a}=(a_1,\dots,a_K)$, let $\sigma$ be a sequence of $[K]$ such that $a_{\sigma(1)}\geq \dots \geq a_{\sigma(K)}$. Then  the WOWA criterion is defined in the following way~\cite{TT11}:
$${\rm wowa}_{(\pmb{v}, \pmb{p})}(\pmb{a})=\sum_{j\in [K]} \omega_j a_{\sigma(j)},$$
where
$$\omega_j=w^*(\sum_{i\leq j} p_{\sigma(i)})-w^*(\sum_{i<j} p_{\sigma(i)}).$$

The value of $\omega_j$ is a weight assigned to number $a_{\sigma(j)}$. It is not difficult to show that $\omega_j\in [0,1]$ for each $j\in [K]$, and $\sum_{j\in [K]} \omega_j=1$. 
Figure~\ref{fig1} shows two boundary cases of the form of function~$w^*$, which are attained at
 $\pmb{v}^1=(1,0,\dots,0)$ and $\pmb{v}^2=(1/K,\dots,1/K)$, respectively. The vector $\pmb{v}^2$ models the weighted mean, i.e. in this case we get ${\rm wowa}_{(\pmb{v}^2,\pmb{p})}(\pmb{a})=\sum_{j\in [K]} p_j a_j$. The vector $\pmb{v}^1$ models the weighted maximum, which in the case of uniform $\pmb{p}=(1/K,\dots,1/K)$ is the usual maximum operator. 
 It is easily seen that
  for arbitrary $\pmb{v}$ and uniform $\pmb{p}=(1/K,\dots,1/K)$, WOWA  becomes the OWA operator.
  An easy computation  shows that
   the WOWA operator is monotone, i.e. when $\pmb{a}$ and $\pmb{a}'$ are such that $a_j\geq a_j'$ for all $j\in [K]$, then ${\rm wowa}_{(\pmb{v},\pmb{p})}(\pmb{a})\geq {\rm wowa}_{(\pmb{v},\pmb{p})}(\pmb{a}')$. 
Since it is a convex combination of the components of $\pmb{a}$, we have $\min_{j\in [K]} a_j \leq {\rm wowa}_{\pmb{v},\pmb{p}}(\pmb{a}) \leq \max_{j\in [K]} a_j$. Additionally, 
 $\sum_{j\in [K]} p_j a_j \leq {\rm wowa}_{(\pmb{v},\pmb{p})}(\pmb{a}) \leq \max_{j\in [K]} a_j$ holds, when
 $v_1\geq\dots\geq v_K$.

We now apply the WOWA operator to the uncertain problem $\mathcal{P}$ and provide the interpretation of the vectors $\pmb{v}$ and $\pmb{p}$. For a given solution $X\in \Phi$, let us define:
$${\rm WOWA}(X)={\rm wowa}_{(\pmb{v},\pmb{p})}(F(X,\pmb{c}_1),\dots,F(X,\pmb{c}_K)).$$
We thus obtain an aggregated value for $X$, by applying the WOWA criterion to the vector of the costs of $X$ under scenarios in $\Gamma$. Given vectors $\pmb{v}$ and $\pmb{p}$, we consider the following optimization problem:
$$ \textsc{Min-Wowa}~\mathcal{P}:\;\; \min_{X\in \Phi} {\rm WOWA}(X).$$

The vector $\pmb{p}=(p_1,\dots,p_K)$ denotes just the probabilities for scenarios. 
 The vector $\pmb{v}$ models the level of risk aversion (or the degree of pessimism/optimism) of  a decision maker. Namely, the more uniform is the weight distribution in $\pmb{v}$ the less risk averse a decision maker is. In particular, $\pmb{v}^2=(1/K,\dots,1/K)$ means that decision maker is risk indifferent and minimizes the expected solution cost. On the other hand, the vector $\pmb{v}^1=(1,0,\dots,0)$ and the uniform vector $\pmb{p}=(1/K,\dots,1/K)$
mean that the decision maker is extremely risk averse and minimizes the solution cost assuming that the worst scenario for the computed solution will occur. In general, vector $\pmb{v}$ allows us to model various attitudes of decision makers towards risk. Moreover, nonincreasing weights are consistent with the concept of robustness. 
Given a solution $X$, let $\sigma$ be such that $F(X,\pmb{c}_{\sigma(1)})\geq \dots \geq F(X,\pmb{c}_{\sigma(K)})$. Then, the value of $\omega_j$ can be seen as a distorted, rank-dependent probability of scenario $\pmb{c}_{\sigma(j)}$, and ${\rm WOWA}(X)$ is the expected solution cost with respect to the distorted probabilities. Notice that $\omega_j$ depends not only on the scenario probability but also on the solution~$X$.

Let us consider again the sample \textsc{Shortest Path} problem shown in Figure~\ref{figex1}. Suppose that $\pmb{v}=(0.5, 0.3, 0.2, 0)$. The computation of the weights $\omega_1,\dots,\omega_4$ for paths $X_1=\{e_1,e_4\}$ and $X_2=\{e_1,e_3,e_5\}$ is shown in Figure~\ref{figex2}. For $X_1$ we get $F(X_1,\pmb{c}_1)\geq F(X_1,\pmb{c}_4)\geq F(X_1,\pmb{c}_2)\geq F(X_1,\pmb{c}_3)$ and $\pmb{\omega}=(0.8,0.08,0.12,0)$. Hence ${\rm WOWA}(X_1)=0.8\cdot10+0.08\cdot2+0.12\cdot 1 +0 \cdot 1=8.28$. Observe that for $X_1$, the worst scenario $\pmb{c}_1$ has the weight equal to 0.8, which is greater than $p_1=0.5$ and the best scenario $\pmb{c}_3$ has the weight equal to 0, which is less than $p_3=0.2$. This example illustrates  how the vector $\pmb{v}$ distorts the scenario probabilities for solution $X_1$, by paying more attention to worse scenarios. In a similar way we compute the weights for path $X_2$, obtaining $\pmb{\omega}=(0.2, 0.36,0.44,0)$ and ${\rm WOWA(X_2)}=0.2\cdot 8+0.36\cdot 7+0.44\cdot 5+0\cdot 5=6.32$. Observe that ${\rm WOWA}(X_2)<{\rm WOWA}(X_1)$, so a risk averse decision maker would prefer solution $X_2$ over $X_1$, contrary to the case when the expected value is used as the criterion of choosing a solution.

\begin{figure}[bht]
\centering
\includegraphics[scale=.5]{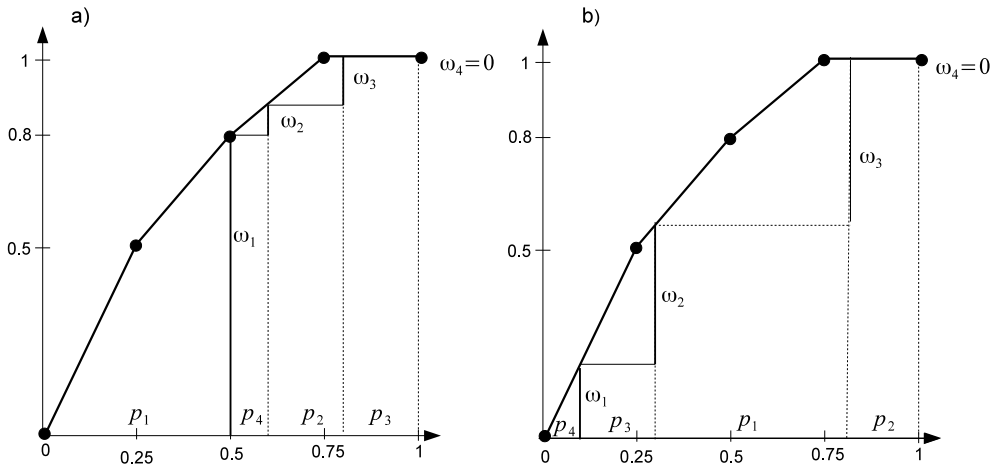}
\caption{The weights $\omega_1,\dots,\omega_4$ for paths a) $X_1=\{e_1,e_4\}$ and b) $X_2=\{e_1,e_3,e_5\}$.}
\label{figex2}       
\end{figure}

\section{Complexity of the problem}
\label{secComp}

In this section we discuss the complexity of \textsc{Min-Wowa}~$\mathcal{P}$.  Notice that \textsc{Min-Wowa}~$\mathcal{P}$  becomes the \textsc{Min-Owa}~$\mathcal{P}$ problem, discussed in~\cite{KZ15}, when $\pmb{p}=(1/K,\dots,1/K)$ and $\pmb{v}$ is an arbitrary weight vector.
 If additionally $\pmb{v}=(1,0,\dots,0)$, then \textsc{Min-Wowa}~$\mathcal{P}$ is the \textsc{Min-Max}~$\mathcal{P}$ problem, widely discussed in the literature devoted to the robust discrete optimization. Hence all negative complexity and approximation results known for \textsc{Min-Owa}~$\mathcal{P}$ and \textsc{Min-Max}~$\mathcal{P}$ remain valid for \textsc{Min-Wowa}~$\mathcal{P}$. Let us recall that
\textsc{Min-Max}~$\mathcal{P}$ is usually NP-hard even when $K=2$. In particular, this is the case for all basic network problems such as \textsc{Shortest Path}, \textsc{Minimum Assignment}, \textsc{Minimum Spanning Tree}, or \textsc{Minimum Cut} (see, e.g.,~\cite{ABV08, AV01, KY97}). Furthermore, when $K$ is  part of the input, then for
 all the aforementioned problems, \textsc{Min-Max}~$\mathcal{P}$ is strongly NP-hard and  also hard to approximate within any constant factor~\cite{KZ09, KZ13}. The problem complexity becomes worse when the maximum criterion is replaced with the more general OWA one. It has been shown in~\cite{KZ15}, that all the basic network problems are then not at all approximable.  This negative result holds when the vector $\pmb{v}$ is arbitrary. However, for nonincreasing weights in $\pmb{v}$ the following positive result is known:
\begin{thm}[\cite{KZ15}]
\label{thmowa}
	When $v_1\geq v_2 \geq \dots \geq v_K$ and $\mathcal{P}$ is polynomially solvable, then \textsc{Min-Owa}~$\mathcal{P}$ is approximable within $v_1K$.
\end{thm}

In the next section we will generalize Theorem~\ref{thmowa} to \textsc{Min-Wowa}~$\mathcal{P}$. It has been 
shown in~\cite{KZ15}, that \textsc{Min-Owa}~$\mathcal{P}$ can be solved in pseudopolynomial time and even admits a
 fully polynomial-time approximation scheme (FPTAS), 
 when $K$ is constant and some additional assumptions for $\mathcal{P}$ are satisfied. We now show that the reasoning  can be easily generalized to \textsc{Min-Wowa}~$\mathcal{P}$. Observe that ${\rm wowa}_{(\pmb{v},\pmb{p})}(\pmb{a})$ is nondecreasing with respect to each $a_j$ in $\pmb{a}$. This fact immediately implies, that there exists an optimal solution $X$ to \textsc{Min-Wowa}~$\mathcal{P}$, which is efficient (Pareto optimal), i.e. for which there is no solution $Y$ such that $F(Y,\pmb{c}_j)\leq F(X,\pmb{c}_j)$ for each $j\in [K]$ with at least one strict inequality. Notice also that each optimal solution to \textsc{Min-Wowa}~$\mathcal{P}$ must be efficient when all components of $\pmb{p}$ and $\pmb{v}$ are positive. Thus it is sufficient to enumerate efficient solutions  and pick up a one, say~$X^*$, with
 the minimum~ ${\rm WOWA(X^*)}$.
 For some problems, for example when $\mathcal{P}$ is
 the   \textsc{Shortest Path} or \textsc{Minimum Spanning Tree} problem, 
 such enumeration of efficient solutions 
 can be done in pseudopolynomial time, provided that $K$ is constant,
 by using techniques given in~\cite{ABV10}. Hence, for constant $K$,
  \textsc{Min-Wowa}~\textsc{Shortest Path} 
  and \textsc{Min-Wowa}~\textsc{Minimum Spanning Tree} 
  can be solved in pseudopolynomial time.

In order to construct an FPTAS, we need a definition of an \emph{exact problem} associated with $\mathcal{P}$ and scenario set $\Gamma$ (see~\cite{MS08}). Given a vector $(b_1,\dots, b_K)$, we ask if there  is a solution $X\in \Phi$ such
  that $F(X,\pmb{c}_j)=b_j$ for all $j\in [K]$. Let us fix $\epsilon>0$ and let $P_{\epsilon}(\Phi)$ be the set of solutions such that for all $X\in \Phi$, there is $Y\in P_{\epsilon}(\Phi)$ such that $F(Y,\pmb{c}_j)\leq (1+\epsilon)\,F(X,\pmb{c}_j)$ for all $j\in [K]$.
  Basing on the results obtained in~\cite{PY00}, it was proven in~\cite{MS08} that
if the exact problem associated with~$\mathcal{P}$ can be solved in pseudopolynomial time, then for any $\epsilon>0$, the set $P_{\epsilon}(\Phi)$ can be determined in time polynomial in the input size and $1/\epsilon$. This implies the following result (the reasoning is the same as in~\cite{KZ15}):
\begin{thm}
\label{thmfptas}
If the exact problem associated with~$\mathcal{P}$ can be solved in pseudopolynomial time, then $\textsc{Min-Wowa}~\mathcal{P}$ admits an FPTAS.
\end{thm}
\begin{proof}
	Let us fix $\epsilon>0$ and let $Y$ be a solution with minimum value of $\mathrm{WOWA}(Y)$ among all the solutions in $P_{\epsilon}(\Phi)$. From the  results obtained in~\cite{MS08, PY00}, it follows that we can find~$Y$ in time polynomial in the input size and $1/\epsilon$. Assume that $X^*$ is an optimal solution to \textsc{Min-Wowa}~$\mathcal{P}$. Define vector $\pmb{b}^*=((1+\epsilon)F(X^*,\pmb{c}_1),\dots,(1+\epsilon)F(X^*,\pmb{c}_K))$. 
By the definition of~$P_{\epsilon}(\Phi)$, there exists a solution~$Y'\in P_{\epsilon}(\Phi)$
such that  $F(Y',\pmb{c}_j)\leq (1+\epsilon)F(X^*,\pmb{c}_j)$ for all $j\in [K]$. The choice of~$Y$
and  the monotonicity of  WOWA  implies ${\rm WOWA(Y)}\leq {\rm WOWA(Y')}\leq {\rm wowa}_{(\pmb{v},\pmb{p})}(\pmb{b}^*)=(1+\epsilon){\rm WOWA}(X^*)$.  We have thus obtained an FPTAS for $\textsc{Min-Wowa}~\mathcal{P}$. 
\end{proof}

It turns out that the exact problem associated with~$\mathcal{P}$ can be solved in pseudopolynomial time for some particular problems~$\mathcal{P}$, provided that 
the number of scenarios~$K$ is constant.
This is the case for \textsc{Shortest Path},  \textsc{Minimum Spanning Tree} and some other problems described, for example, in~\cite{ABV10}. However,
it is worth pointing out that
  the running time of the obtained FPTAS's  is exponential in $K$, so their practical applicability is limited to very small values of $K$. In the next section we will construct an approximation algorithm, which can be applied for larger values of $K$.

\section{Approximation algorithm}
\label{secAppr}

In this section we construct an approximation algorithm for \textsc{Min-Wowa}~$\mathcal{P}$ under the assumptions that $v_1\geq v_2\geq \dots \geq v_K$ and $\mathcal{P}$ is polynomially solvable. We  will also assume that $p_j>0$ for each $j\in [K]$. When $p_j=0$ for some $j\in [K]$, then we can remove scenario $\pmb{c}_j$ from $\Gamma$ without affecting the problem. We first prove some properties of the WOWA operator. 
Let $\pmb{a}=(a_1,\dots,a_K)$ be a vector of nonnegative numbers. Let $\pi$ be any sequence of $[K]$. Let us define
$$f_{\pi}(\pmb{a})= \sum_{j\in [K]} \omega_j a_{\pi(j)},$$ 
where $\omega_j=w^*(\sum_{i\leq j} p_{\pi(i)})-w^*(\sum_{i<j} p_{\pi(i)})$ and $w^*$ is the piecewise linear  function induced by the vector of weights $\pmb{v}$ (as in the definition of the WOWA operator). Observe that $f_{\pi}(\pmb{a})={\rm wowa}_{(\pmb{v},\pmb{p})}(\pmb{a})$ when  the sequence $\pi$ is such that $a_{\pi(1)}\geq \dots \geq a_{\pi(K)}$. The following lemma expresses the intuitive fact that $f_{\pi}(\pmb{a})$ is a lower bound on ${\rm wowa}_{(\pmb{v},\pmb{p})}(\pmb{a})$.

\begin{lem}
\label{lem1}
Given any vector $\pmb{a}=(a_1,\dots,a_K)$ and any sequence $\pi$ of $[K]$. Then
${\rm wowa}_{(\pmb{v},\pmb{p})}(\pmb{a})\geq f_{\pi}(\pmb{a})$.
\end{lem}
\begin{proof}
Assume w.l.o.g. that $a_1\geq a_2\geq \dots\geq a_K$.
Consider two neighbor elements $a_{\pi(i)}$ and $a_{\pi(i+1)}$	in $\pi$ such that $a_{\pi(i)}\leq a_{\pi(i+1)}$.  Let us interchange $a_{\pi(i)}$ and $a_{\pi(i+1)}$ in $\pi$ and denote the resulting sequence by $\pi'$. We will show that $f_{\pi'}(\pmb{a})\geq f_{\pi}(\pmb{a})$, where the equality holds when $a_{\pi(i)}=a_{\pi(i+1)}$. This will complete the proof since we can transform $\pi$ into  $\sigma=(1,\dots,K)$ by using a finite number of such element interchanges without decreasing the value of $f_{\pi}$ and $f_{\sigma}(\pmb{a})={\rm wowa}_{(\pmb{v},\pmb{p})}(\pmb{a})$. 
 It is easily seen that $f_{\pi'}(\pmb{a})-f_{\pi}(\pmb{a})=\omega'_i a_{\pi(i+1)} + \omega'_{i+1} a_{\pi(i)}-\omega_i a_{\pi(i)} - \omega_{i+1} a_{\pi(i+1)}=(\omega'_{i+1}-\omega_{i})a_{\pi(i)}-(\omega_{i+1}-\omega'_i)a_{\pi(i+1)}$. 
 Equality $\omega'_i+\omega'_{i+1}=\omega_i+\omega_{i+1}$ (see Figure~\ref{fig2}a) holds, and so $\omega'_{i+1}-\omega_{i}=\omega_{i+1}-\omega'_{i}=\alpha$. Hence $f_{\pi'}(\pmb{a})-f_{\pi}(\pmb{a})=\alpha(a_{\pi(i)}-a_{\pi(i+1)})$. Since $w^*$ is concave, we conclude that $\omega_{i+1}/p_{\pi(i+1)}\leq \omega'_i/p_{\pi(i+1)}$, which yields $\alpha\leq 0$ since $p_{\pi(i+1)}>0$. Hence $f_{\pi'}(\pmb{a})\geq f_{\pi}(\pmb{a})$.
\end{proof}
\begin{figure}[ht]
\centering
\includegraphics[scale=.7]{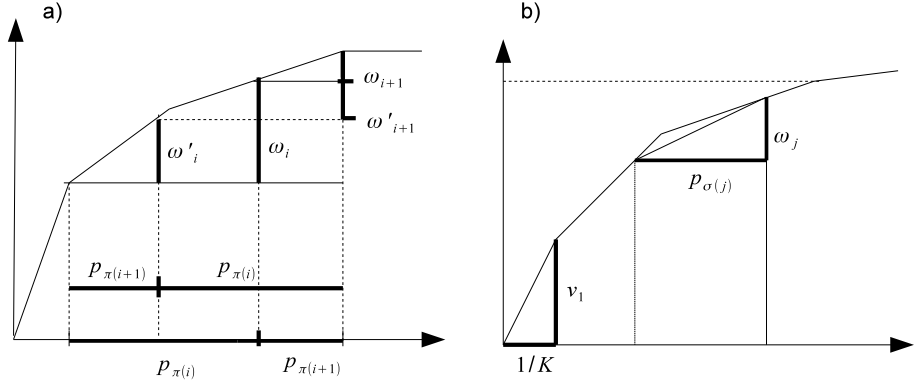}
\caption{Illustrations of the proofs of Lemma~\ref{lem1} and Lemma \ref{lem2}.}
\label{fig2}      
\end{figure}
\begin{lem}
\label{lem2}
	Given any vector $\pmb{a}=(a_1,\dots,a_K)$. Then ${\rm wowa}_{(\pmb{v},\pmb{p})}(\pmb{a})\leq v_1K\sum_{j\in [K]} p_ja_j$.
\end{lem}
\begin{proof}
Since $w^*$ is concave and piecewise linear, it follows that
$\frac{\omega_j}{p_{\sigma(j)}}\leq \frac{v_1}{1/K}=v_1K$ for each $j\in [K]$ (see Figure~\ref{fig2}(b)). In consequence, ${\rm wowa}_{(\pmb{v},\pmb{p})}(\pmb{a})=\sum_{j\in [K]} \omega_j a_{\sigma(j)}\leq \sum_{j\in [K]} v_1Kp_{\sigma(j)}a_{\sigma(j)}=v_1K\sum_{j\in [K]} p_j a_j$.
\end{proof}

Let $\hat{c}_i={\rm wowa}_{(\pmb{v},\pmb{p})}(c_{1i},\dots,c_{Ki})$ be the aggregated cost of element $e_i\in E$ over all scenarios. Let $\hat{X}$ be an optimal solution for the costs $\hat{c}_i$, $i\in [n]$. We begin with a general result:
\begin{thm}
\label{thm1}
 Given any $X$. Then ${\rm WOWA}(\hat{X})\leq Kv_1\cdot {\rm WOWA}(X)$.
\end{thm}
\begin{proof}
Let $\sigma$ be a sequence of $[K]$ such that $F(\hat{X},\pmb{c}_{\sigma(1)})\geq\dots\geq F(\hat{X},\pmb{c}_{\sigma(K)})$ and $\omega_j=w^*(\sum_{i\leq j} p_{\sigma(i)})-w^*(\sum_{i<j} p_{\sigma(i)})$. The definition of the WOWA operator and Lemma~\ref{lem1} imply the following inequality:
\begin{equation}
\label{f1} 
\mathrm{WOWA}(\hat{X})=\sum_{j\in [K]} \omega_j \sum_{e_i\in \hat{X}} c_{\sigma(j)i}=
\sum_{e_i\in \hat{X}} \sum_{j\in [K]} \omega_j c_{\sigma(j)i}\leq \sum_{e_i\in \hat{X}} \hat{c}_i.
\end{equation} 	
Using Lemma~\ref{lem2}, we get
$
	\hat{c}_i \leq v_1 K \sum_{j\in [K]} p_jc_{ji}.
$
Hence, from the definition of $\hat{X}$,  we obtain
\begin{equation}
\label{f2}
\sum_{e_i\in \hat{X}} \hat{c}_i \leq \sum_{e_i\in X}\hat{c}_i\leq Kv_1\sum_{e_i\in X} \sum_{j\in [K]}  p_j c_{ji}.
\end{equation} 
Since $v_1\geq \dots \geq v_K$ it follows that
\begin{equation}
\label{f3}
\mathrm{WOWA}(X)\geq \sum_{j\in [K]} p_j F(X,\pmb{c}_j)=\sum_{j\in [K]} p_j\sum_{e_i \in X} c_{ji}=\sum_{e_i\in X}\sum_{j\in [K]} p_jc_{ji}.
\end{equation}
Combining~(\ref{f1}), (\ref{f2}) and~(\ref{f3}) completes the proof.
\end{proof}
Theorem~\ref{thm1} leads to the following corollary:
\begin{cor}
\label{cor1}
	If $v_1\geq\dots\geq v_K$ and $\mathcal{P}$ is polynomially solvable, then \textsc{WOWA}~$\mathcal{P}$ is approximable within $v_1K$.
\end{cor}

The bound obtained in Corollary~\ref{cor1} is tight and the worst case instance for the approximation algorithm is the same as the one shown in~\cite{KZ15}. Observe that the approximation ratio depends on the weight distribution in $\pmb{v}$. The more uniform is the weight distribution  the smaller is the approximation ratio. We get the largest approximation ratio equal to $K$, when WOWA is the weighted maximum. On the other hand, when $v_1=1/K$, i.e. when WOWA is the expected value, then we get an exact polynomial time algorithm for the problem.

In many cases the deterministic problem $\mathcal{P}$ is NP-hard, but is approximable within a factor of~$\gamma$. In this case the following result can be established.
\begin{thm}
	If $v_1\geq\dots\geq v_K$  and  $\mathcal{P}$ is approximable within $\gamma$, then \textsc{Min-Wowa}~$\mathcal{P}$ is approximable within $\gamma v_1K$.
\end{thm}	
\begin{proof}
	The proof is similar to the proof of Theorem~\ref{thm1}. In order to get a solution for costs $\hat{c}_i$ a $\gamma$-approximation algorithm is applied. It is then enough to modify inequality~(\ref{f2}), so that $\sum_{e_i\in \hat{X}} \hat{c}_i \leq \gamma \sum_{e_i\in X}\hat{c}_i\leq \gamma Kv_1\sum_{e_i\in X} \sum_{j\in [K]}  p_j c_{ji}$. The rest of the proof is the same.
\end{proof}

\section{Mixed integer programming formulation}
\label{secMIP}

In this section we design a mixed integer programming (MIP) formulation for \textsc{Min-Wowa}~$\mathcal{P}$. We will use the idea proposed in~\cite{OS09} (see also~\cite{CG15,FPP14,OS03} for alternative formulations for
the OWA operator). Let us associate a binary variable $x_i\in \{0,1\}$ with each element $e_i\in E$. Let $\chi(\Phi)\subseteq \{0,1\}^n$ be the set of all characteristic vectors of $\Phi$. Each  vector $\pmb{x}=(x_1,\dots,x_n)\in \chi(\Phi)$ defines a feasible solution $X$ such that $e_i\in X$ if and only if $x_i=1$. We will assume that $\chi(\Phi)$ can be described by a set of linear constraints involving variables $x_1,\dots,x_n$. 
From now on we will identify a feasible solution $X\in \Phi$ with the corresponding characteristic vector $\pmb{x}\in \chi(\Phi)$.
Let us fix a feasible solution $\pmb{x}\in \chi(\Phi)$. Let $\sigma$ be such that $F(\pmb{x},\pmb{c}_{\sigma(1)})\geq \dots \geq F(\pmb{x},\pmb{c}_{\sigma(K)})$. Define vector $\pmb{\alpha}=(\alpha_0, \alpha_1,\dots,\alpha_K)$ such that $\alpha_i=\sum_{j\leq i} p_{\sigma(j)}$, and $\alpha_0=0$.
Let us define $h_{\pmb{x}}(\theta)=F(\pmb{x},\pmb{c}_{\sigma(i)})$ for $\alpha_{i-1}<\theta\leq \alpha_i$, $i\in [K]$, $\theta\in (0,1]$. The following equality holds~\cite{OS09}:
\begin{equation}
\label{lemn1}
{\rm WOWA}(\pmb{x})=K\sum_{j\in [K]} v_j \int_{\frac{j-1}{K}}^{\frac{j}{K}}h_{\pmb{x}}(\theta) \mathrm{d\,}\theta.
\end{equation}

\begin{figure}[ht]
\centering
\includegraphics[scale=.7]{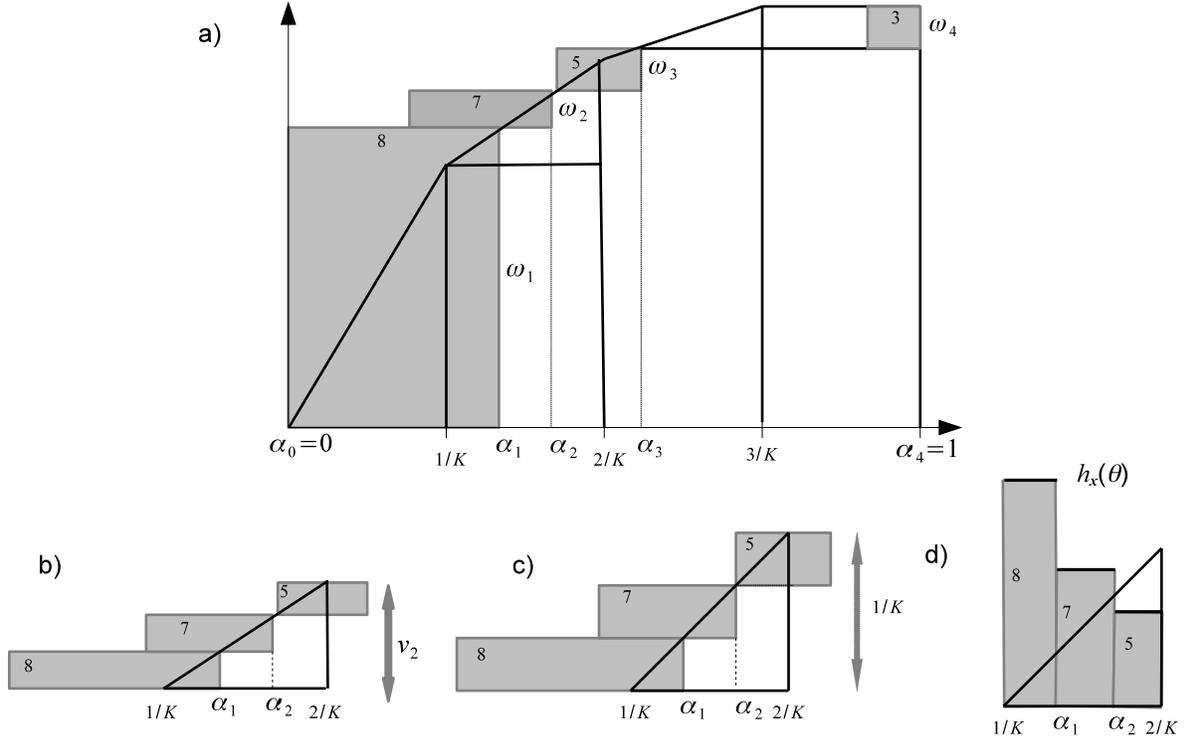}
\caption{Illustration of formula~(\ref{lemn1}) for vector $(8,5,7,3)$. In (c) the right-hand edge of the triangle in (b) is scaled by the factor $1/(Kv_2)$.}
\label{figex3}       
\end{figure}

Formula~(\ref{lemn1}) is illustrated in Figure~\ref{figex3}. Observe first that the value of ${\rm WOWA} (\pmb{x})$ is equal to the size of the area of the grey rectangles in Figure~\ref{figex3}a. In Figure~\ref{figex3}b the portion of this area that touches the triangle with base $[1/K, 2/K]$ is shown. After scaling the right-hand edge of this triangle by the factor $1/(v_2K)$ we obtain the area shown in Figure~\ref{figex3}c. After rotating the rectangles we obtain the area shown in Figure~\ref{figex3}d. Now it is easy to see that the size of this area is equal to $\int_{\frac{1}{K}}^{\frac{2}{K}}h_{\pmb{x}}(\theta)\mathrm{d\,}\theta$. Multiplying it by $v_2K$, we get the size of the area from Figure~\ref{figex3}b.

Equality~(\ref{lemn1}) has the following interpretation (see also~\cite{OS09}). The value of $K\int_{\frac{j-1}{K}}^{\frac{j}{K}}h_{\pmb{x}}(\theta)\mathrm{d\,}\theta$ is the average within the $j$th portion of $1/K$ largest solution costs. Then ${\rm WOWA} (\pmb{x})$ can be seen as the value of the OWA operator applied to these averages. When $p_j=1/K$, $j\in [K]$, then $K\int_{\frac{j-1}{K}}^{\frac{j}{K}}h_{\pmb{x}}(\theta)\mathrm{d\,}\theta=F(\pmb{x},\pmb{c}_{\sigma(j)})$ and ${\rm WOWA}(\pmb{x})$ becomes the OWA aggregation operator.

 Let us rewrite~(\ref{lemn1}) as follows:  
$$
{\rm WOWA}(\pmb{x})=K\sum_{j\in [K]} v_j \left(\int_{0}^{\frac{j}{K}}h_{\pmb{x}}(\theta)\mathrm{d\,}\theta-\int_{0}^{\frac{j-1}{K}}h_{\pmb{x}}(\theta)\mathrm{d\,}\theta\right)=
$$
$$=K \left(\sum_{j=1}^K v_j \int_{0}^{\frac{j}{K}}h_{\pmb{x}}(\theta)\mathrm{d\,}\theta - \sum_{j=0}^{K-1} v_{j+1}\int_{0}^{\frac{j}{K}}h_{\pmb{x}}(\theta)\mathrm{d\,}\theta\right).$$
Define $v_{K+1}=0$. Since  $\int_{0}^{0}h_{\pmb{x}}(\theta)\mathrm{d\,}\theta=0$, we have
$${\rm WOWA}(\pmb{x})=K \left(\sum_{j=1}^K v_j \int_{0}^{\frac{j}{K}}h_{\pmb{x}}(\theta)\mathrm{d\,}\theta - \sum_{j=1}^{K} v_{j+1}\int_{0}^{\frac{j}{K}}h_{\pmb{x}}(\theta)\mathrm{d\,}\theta\right),$$
and we get the following equality:
\begin{equation}
\label{formw}
	{\rm WOWA}(\pmb{x})=K\sum_{j\in [K]} (v_j-v_{j+1}) \int_{0}^{\frac{j}{K}}h_{\pmb{x}}(\theta)\mathrm{d\,}\theta.
\end{equation}
Let us denote $L_j(\pmb{x})=\int_{0}^{\frac{j}{K}}h_{\pmb{x}}(\theta)\mathrm{d\,}\theta$ and $v'_j=v_j-v_{j+1}$, $j\in [K]$. Observe that $v'_j\geq 0$ for all $j\in [K+1]$, by the assumption that $v_1\geq v_2 \geq \dots \geq v_K$.
We are now ready to design a MIP formulation. In order to do this we adopt the idea from~\cite{OS09}. Observe first that the value of $L_j(\pmb{x})$ for a fixed $\pmb{x}$ can be computed by solving the following linear programming problem:
\begin{equation}
\label{mip1a}
	\begin{array}{llll}
		\max &  \sum_{k\in [K]} z_k F(\pmb{x},\pmb{c}_k) \\
			& \sum_{k\in [K]} z_k=\frac{j}{K}\\
			&0\leq z_k \leq p_k & k\in [K]
	\end{array}
\end{equation}
Indeed, $L_j(\pmb{x})$ can be computed in a greedy way. Let $\sigma$ be such that 
$F(\pmb{x},\pmb{c}_{\sigma(1)})\geq \cdots \geq F(\pmb{x},\pmb{c}_{\sigma(K)})$. We first allocate to the interval $[0,j/K]$ the largest possible portion of $p_{\sigma(1)}$, then the largest possible portion of $p_{\sigma(2)}$ etc., until $[0,j/K]$ is completely filled. This is equivalent to solving~(\ref{mip1a}).
The dual to~(\ref{mip1a}) for a fixed $\pmb{x}$ and $j$ takes the following form:
\begin{equation}
\label{mip1b}
	\begin{array}{llll}
		\min &  \frac{j}{K}\beta_j +\sum_{i\in [K]} p_i\alpha_{ij} \\
		& \beta_j+\alpha_{ij}\geq F(\pmb{x},\pmb{c}_i) & i \in [K]\\
		& \alpha_{ij}\geq 0 & i\in [K]
	\end{array}
\end{equation}
The strong duality theorem implies that $L_j(\pmb{x})$ equals the optimal objective value of~(\ref{mip1b}). Using~(\ref{formw}) and~(\ref{mip1b}) we get that \textsc{Min-Wowa}~$\mathcal{P}$ is equivalent to the following problem:
$$
	\begin{array}{lll}
		\min & K\cdot \sum_{j\in [K]} v_j' (\frac{j}{K} \beta_j + \sum_{i\in [K]} p_i\alpha_{ij}) \\
		& \beta_j+\alpha_{ij}\geq F(\pmb{x},\pmb{c}_i) & i \in [K], j\in [K]\\
		& \alpha_{ij}\geq 0 & i\in [K], j\in [K] \\
		& (x_1,\dots,x_n)\in \chi(\Phi)
	\end{array}
$$
We obtain a MIP formulation by substituting $F(\pmb{x},\pmb{c}_i)=\sum_{k\in [n]} x_i c_{ik}$ and replacing the expression $(x_1,\dots,x_n)\in \chi(\Phi)$ with a system of linear constraints involving $x_1,\dots,x_n$. In the next section we will apply the MIP formulation to a sample problem.

\section{Computational tests}
\label{secExp}
In this section we present the results of some computational tests. The tests were performed for two particular discrete optimization problems, namely the \textsc{Selection} and \textsc{Assignment} problems. We wish to verify the following two questions:
\begin{enumerate}
	\item How efficient is the MIP formulation, i.e. how the computation time, required to solve the MIP model, depends on the number of elements~$n$ in the set~$E$, the number of scenarios~$K$ and the weight distribution in $\pmb{v}$?
	\item What is the quality of the approximation algorithm designed in Section~\ref{secAppr}?
\end{enumerate}

For both problems we used the following method of generating the tested instances.
 For each scenario $j\in [K]$ we chose~$a_j$, which is
 a random integer uniformly distributed in $[1,100]$, and
  then set $p_j=a_j/\sum_{i\in [K]} a_i$, $j\in [K]$, obtaining a positive probability for each scenario.
  In order to fix the weights $v_1,\dots,v_K$ for scenarios, we used  generating function $g_{\alpha}(z)=\frac{1}{1-\alpha}(1-\alpha^z)$ where $\alpha\in (0,1)$ is a fixed parameter. Notice that $g_{\alpha}(z)$ is concave and is such that $g_{\alpha}(0)=0$, $g_{\alpha}(1)=1$. Given $\alpha$ and $K$, we set $v_j=g_{\alpha}(j/K)-g_{\alpha}((j-1)/K)$ for $j\in [K]$ (see Figure~\ref{figalpha}). 
  The value of $\alpha$ expresses an attitude of the decision maker towards  risk.
 The smaller  the value of $\alpha$, the more risk averse the decision maker is
  (the less uniform is the weight distribution $v_1,\dots, v_K$). 
  For each generated instance the CPLEX 12.5 solver with standard settings was used to solve the corresponding MIP formulation. We  fixed the time limit to 3600 seconds. The solver was executed on a computer equipped with  a 2.5 GHz processor with 8~GB RAM.
\begin{figure}[bht]
\centering
\includegraphics[scale=.5]{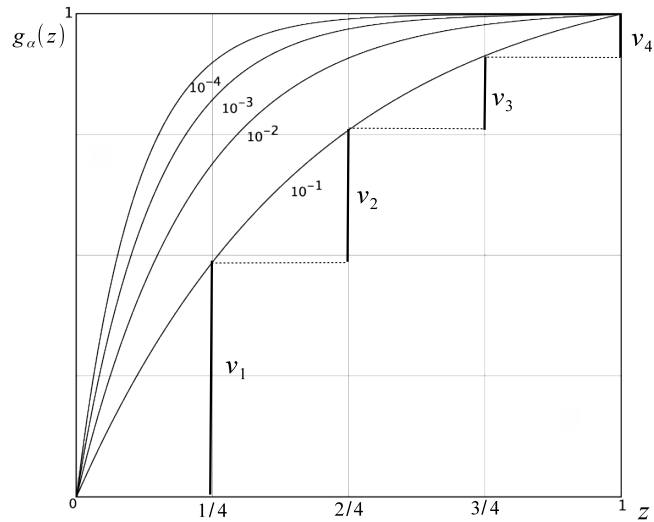}
\caption{Generating function $g_{\alpha}(z)$ for $\alpha\in\{10^{-1}, 10^{-2}, 10^{-3}, 10^{-4}\}$ and the weights $v_1, \dots, v_4$ for $K=4$ and $\alpha=10^{-1}$.}
\label{figalpha}      
\end{figure}

\subsection{The selection problem}

In this section we apply the MIP formulation and the approximation algorithm designed in Section~\ref{secAppr} to the following \textsc{Selection} problem. Assume that $E$ is a set of $n$ items and we wish to choose exactly $q$ of them to minimize the total cost. Hence $\Phi=\{X\subseteq E: |X|=q\}$. The set of characteristic vectors $\chi(\Phi)$ can be described by one constraint of the form $x_1+\dots +x_n=q$, where $x_1,\dots, x_n$ are binary variables associated with the items in $E$. The \textsc{Selection} problem has been recently discussed in a number of papers. Its min-max version has been proven to be NP-hard for two scenarios~\cite{AV01}, strongly NP-hard and hard to approximate within any constant factor when the number of scenarios is  part of the input~\cite{KZ13}.  Hence the same negative results hold for \textsc{Min-Wowa}~$\mathcal{P}$. 

We performed the tests for 
the number of items~$n$ chosen from the set
$\{160, 200\}$, the number of scenarios~$K$
chosen from the set 
$\{2,\dots,20\}$, and the parameter~$\alpha$  chosen from the set $\{10^{-2}, 10^{-3}, 10^{-4}\}$. We also fixed $q=0.25n$, i.e. we assumed that exactly 25\% of the items must be chosen. Under each scenario the cost of item $e_i$ is an integer that is chosen randomly with uniform distribution from the set $\{0,\dots, 100\}$. For each combination of $n$, $K$ and $\alpha$ we have generated 10 random instances.  We first applied the MIP formulation to obtain the optimal solutions for the instances. The computational times required by CPLEX to solve them are shown in Figures~\ref{fig1ex} and~\ref{fig2ex}. It can be observed that the computational times quickly grow with the number of scenarios. The problem is also harder to solve for smaller values of $\alpha$. We were unable to solve any instance with $n=200$, $K=20$, and $\alpha=10^{-4}$ within the time limit of 3600~s.
\begin{figure}[ht]
\centering
\includegraphics[scale=.6]{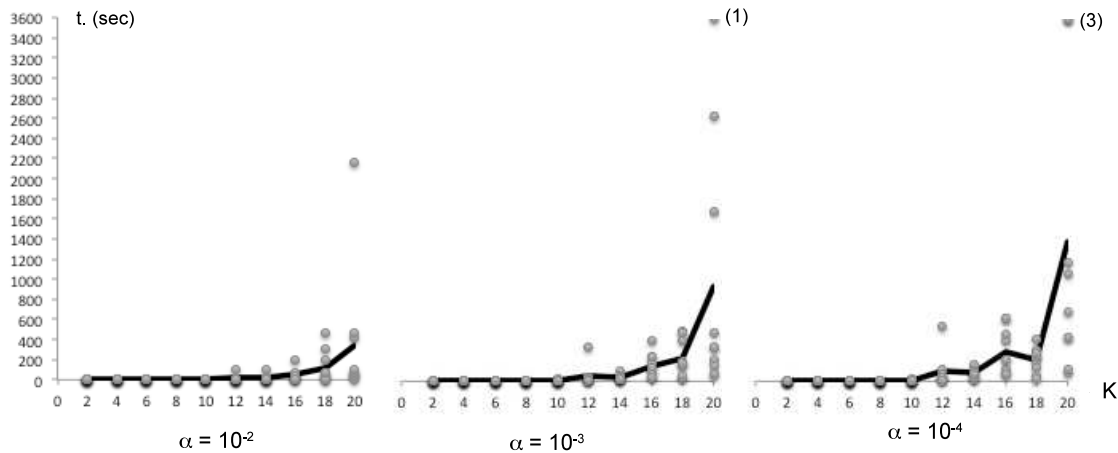}
\caption{
Computational results for  \textsc{Min-WOWA Selection} --
running times for $n=160$ and all combinations of $\alpha$ and $K$. The solid line shows the average computational time. The numbers in brackets show the number of instances which were not solved within 3600~s.}
\label{fig1ex}       
\end{figure}
\begin{figure}[ht]
\centering
\includegraphics[scale=.6]{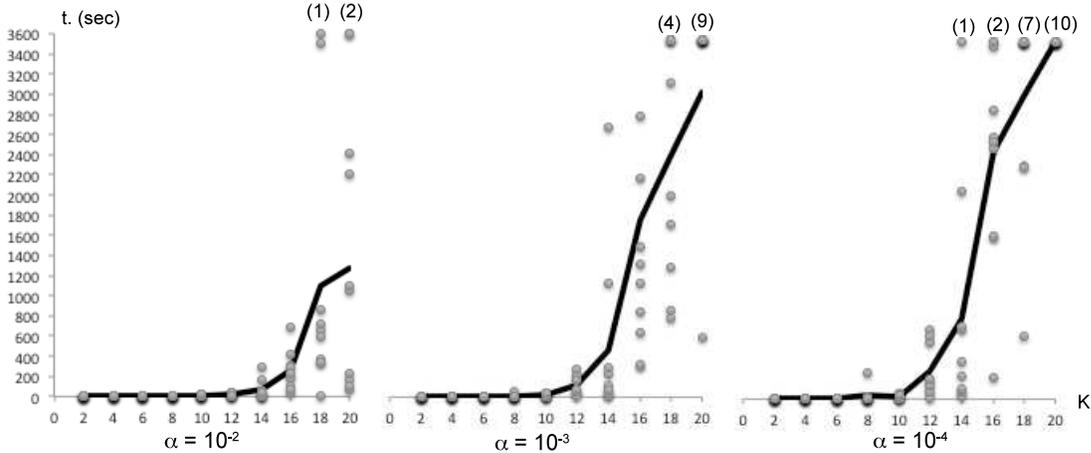}
\caption{Computational results for  \textsc{Min-WOWA Selection} -- running times for $n=200$ and all combinations of $\alpha$ and $K$. The solid line shows the average computational time. The numbers in brackets show the number of instances which were not solved within 3600~s.}
\label{fig2ex}       
\end{figure}

We next applied the approximation algorithm, designed in Section~\ref{secAppr}, to the generated instances. The obtained results are shown in Figures~\ref{fig3ex} and~\ref{fig4aex}.

\begin{figure}[ht]
\centering
\includegraphics[scale=.6]{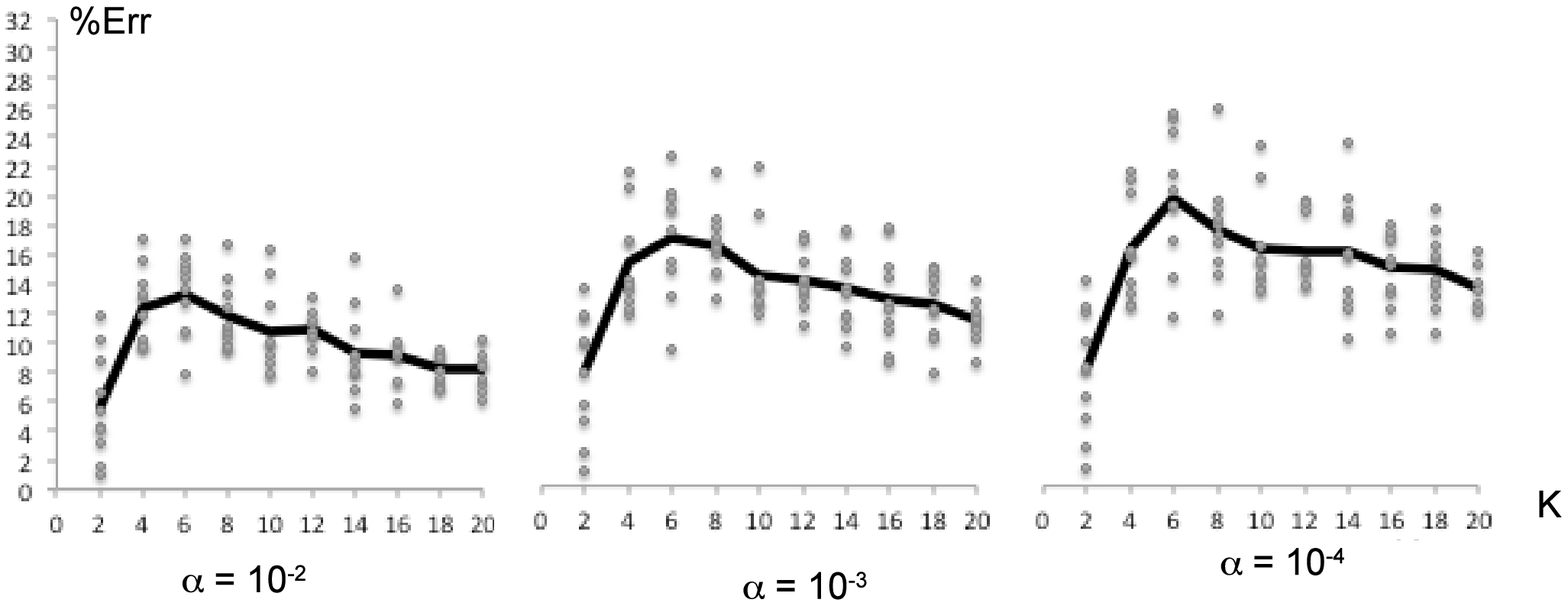}
\caption{The performance of the approximation algorithm 
 run on the instances of
 \textsc{Min-WOWA Selection} --
percentage deviations from the optimum for $n=160$ and all combinations of $\alpha$ and $K$. The solid line shows the average deviation.}
\label{fig3ex}      
\end{figure}
\begin{figure}[ht]
\centering
\includegraphics[scale=.6]{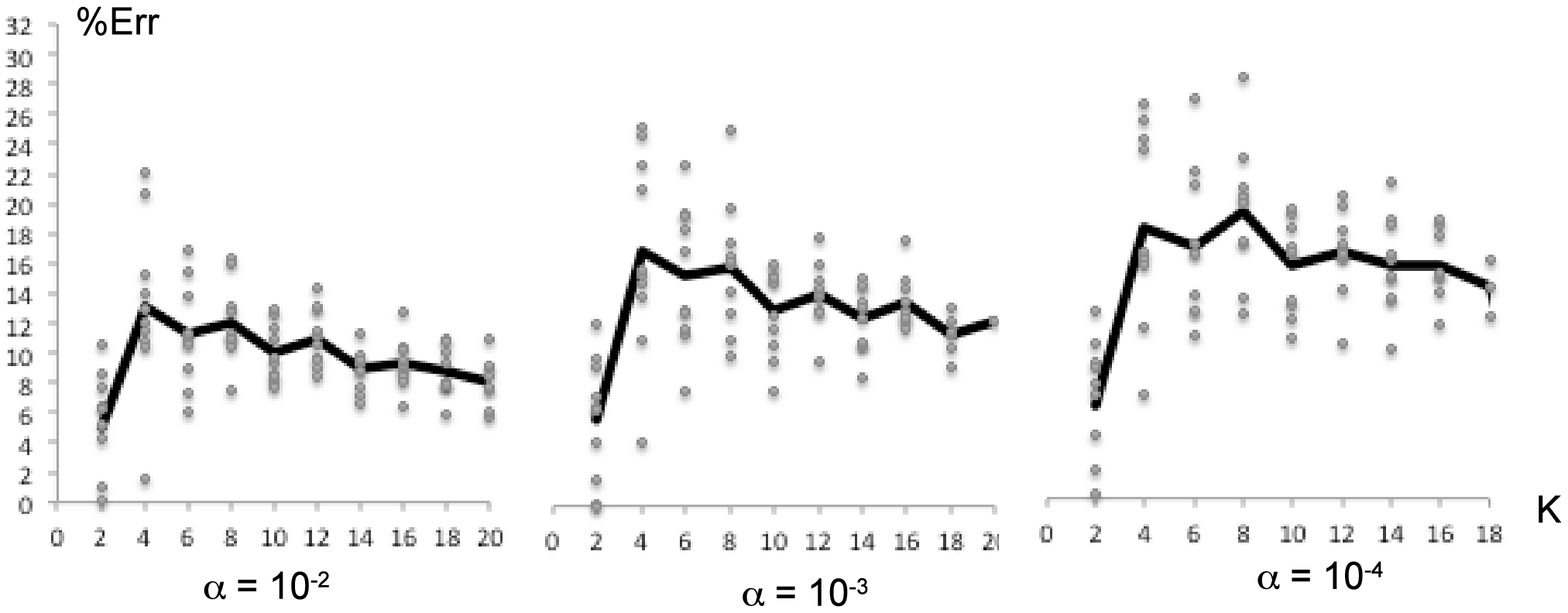}
\caption{The performance of the approximation algorithm 
 run on the instances of
 \textsc{Min-WOWA Selection} --
percentage deviations from the optimum for $n=200$ and all combinations of $\alpha$ and $K$. The solid line shows the average deviation.}
\label{fig4aex}     
\end{figure}

For each instance, for which an optimal solution was known, we computed the percentage deviation of the cost of the approximate solution from the optimum. The quality of the solutions returned by the approximation algorithm seems to be good in comparison with the worst theoretical performance.  The largest reported deviations from optimum are not greater than 32\%, whereas the largest theoretical deviation varies from about 90\% (for $K=2$ and $\alpha=10^{-2}$) to about 638\% (for $K=20$ and $\alpha=10^{-4}$). We can thus obtain reasonable solutions as long as the distribution of the costs under scenarios is uniform. It is interesting that the deviation from the optimum depends more on $\alpha$ than on $K$. The performance of the approximation algorithm is clearly better for larger $\alpha$, i.e when WOWA is closer to the expected value. On the other hand, for a fixed $\alpha$, the performance is significantly better only for $K=2$. For $K=4$ it becomes  worse. Interestingly, one can observe a slightly better performance when $K$ increases from~4 to~20.

\subsection{The assignment problem}

In this section we apply the MIP formulation and the approximation algorithm designed in Section~\ref{secAppr} to the following \textsc{Assignment} problem. We are given a bipartite network $G=(V_1\cup V_2, E)$, where $V_1 \cap V_2=\emptyset$, $|V_1|=|V_2|=m$, and $E=V_1 \times V_2$, $|E|=n=m^2$. Set $\Phi$ contains all subsets of $E$ which form a perfect matching (assignment) in $G$. We can associate a binary variable $x_{ij}$ with each element
(edge) $e_{ij}\in E$ and the set of characteristic vectors $\chi(\Phi)$ can be described by the assignment constraints of the form $\sum_{i\in [m]} x_{ij}=1$ for $j\in [m]$ and $\sum_{j\in [m]} x_{ij}=1$ for $i\in [m]$. The min-max version of the \textsc{Assignment} problem is known to be strongly NP-hard 
and   hard to approximate within any constant factor~\cite{KZ09}.

We performed the tests for
the number of nodes~$m$ ($|V_1|=|V_2|=m$)  chosen from the set
 $\{40, 50\}$, the number of scenarios~$K$
chosen from the set 
 $\{2,\dots,20\}$, and the parameter~$\alpha$  chosen from the set 
 $\{10^{-2}, 10^{-3}, 10^{-4}\}$. Observe that the cardinality of $E$ was 1600 and 2500, respectively. Under each scenario the cost of element~$e_{ij}$ is an integer that is chosen randomly with uniform distribution from the set $\{0,\dots, 100\}$. For each combination of $m$, $K$ and $\alpha$ we have generated 10 random instances.  We first applied the MIP formulation to obtain the optimal solutions for the instances. The computational times required by CPLEX to solve them are shown in Figures~\ref{fig4ex} and~\ref{fig5ex}. One can notice that the computational times quickly grow with the number of scenarios. The problem is also harder to solve for smaller values of $\alpha$. For  $K=10$ and $m=50$ some instances could not be solved within the time limit of 3600~s.
\begin{figure}[ht]
\centering
\includegraphics[scale=.6]{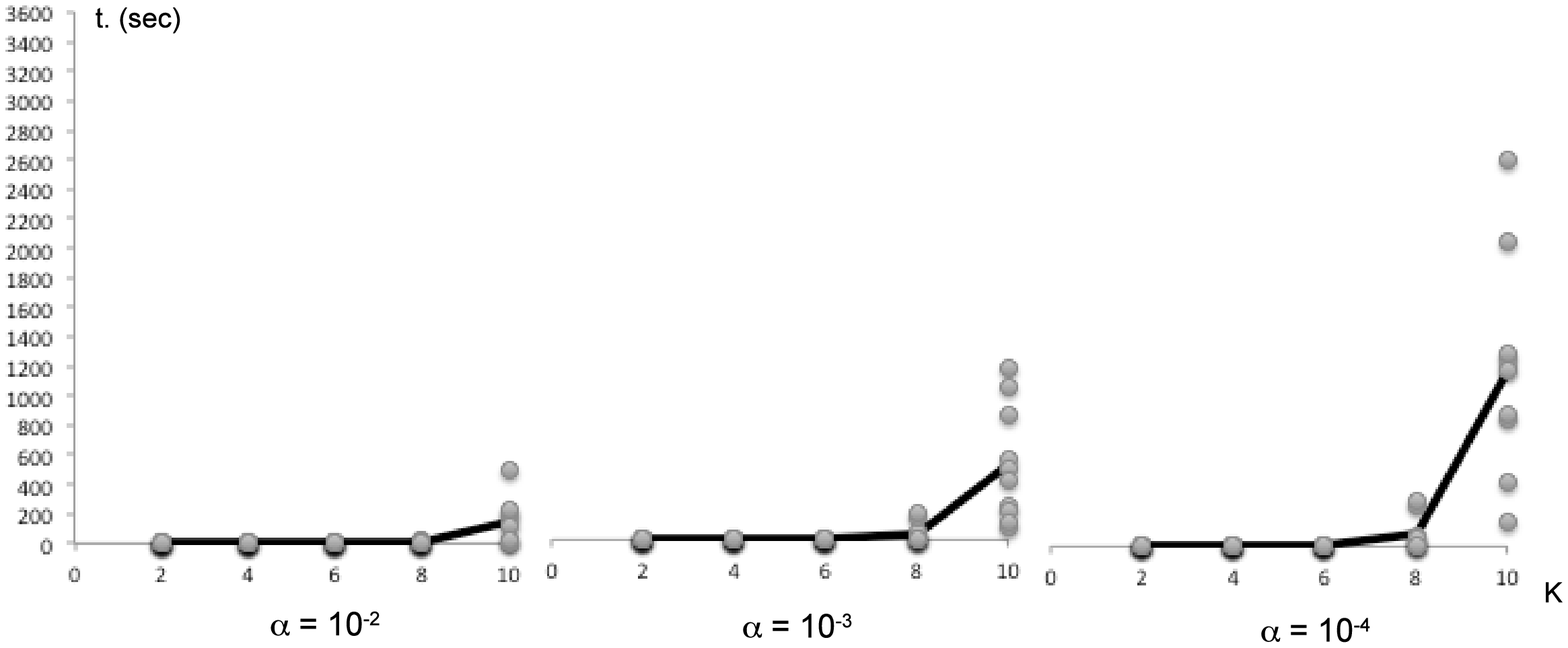}
\caption{Computational results for  \textsc{Min-WOWA  Assignment} --
running times for $m=40$ and all combinations of $\alpha$ and $K$. The solid line shows the average computational time.}
\label{fig4ex}       
\end{figure}
\begin{figure}[ht]
\centering
\includegraphics[scale=.6]{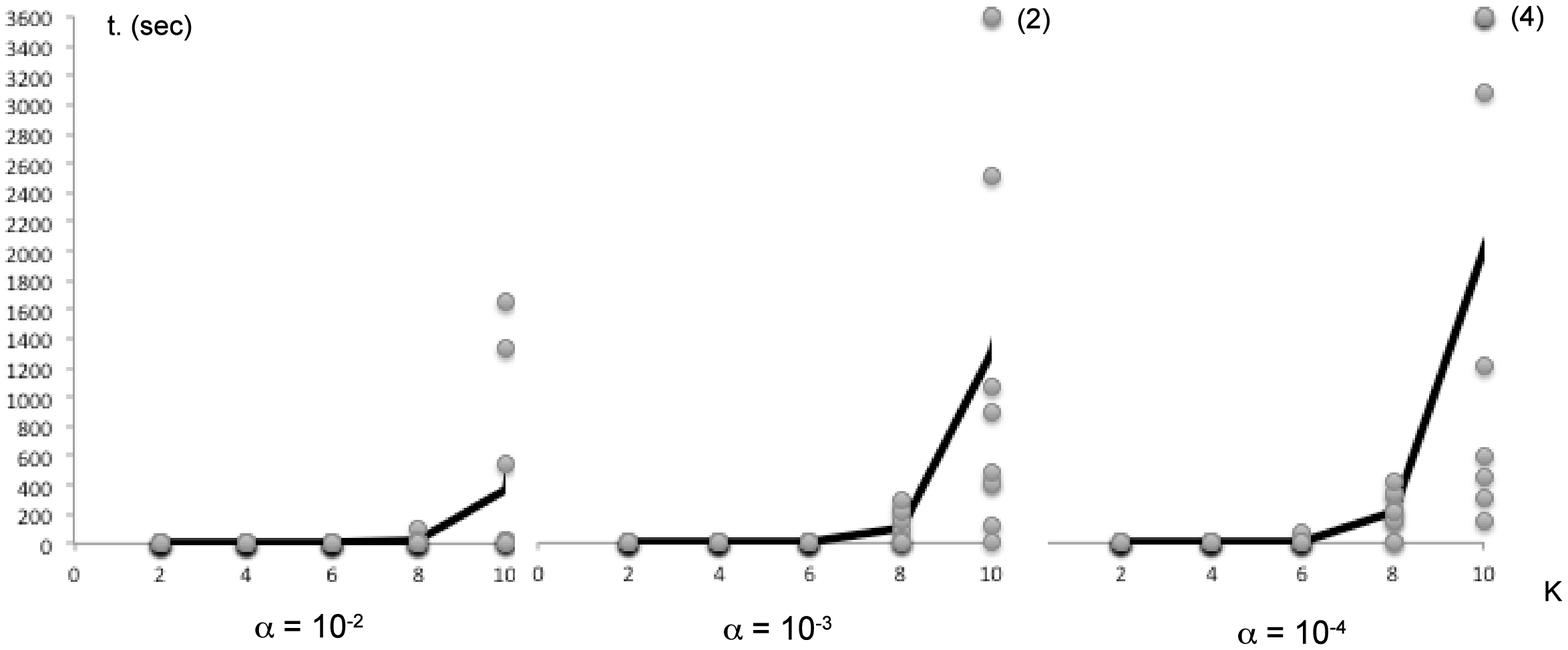}
\caption{Computational results for  \textsc{Min-WOWA Assignment} --
running times for $m=50$ and all combinations of $\alpha$ and $K$. The solid line shows the average computational time. The numbers in brackets show the number of instances which were not solved within 3600~s.}
\label{fig5ex}   
\end{figure}

We next applied the approximation algorithm, constructed in Section~\ref{secAppr}, to the generated instances. The computational results are presented in Figures~\ref{fig6ex} and~\ref{fig7ex}. We can derive similar conclusions as for the \textsc{Selection} problem. The largest deviation from optimum reported was about 30\%, which is much less than the worst theoretical performance which varies from about 90\% (for $K=2$ and $\alpha=10^{-2}$) to about 501\% (for $K=10$ and $\alpha=10^{-4}$). As for the \textsc{Selection} problem, the deviations from optimum is significantly smaller only for $K=2$.
\begin{figure}[ht]
\centering
\includegraphics[scale=.6]{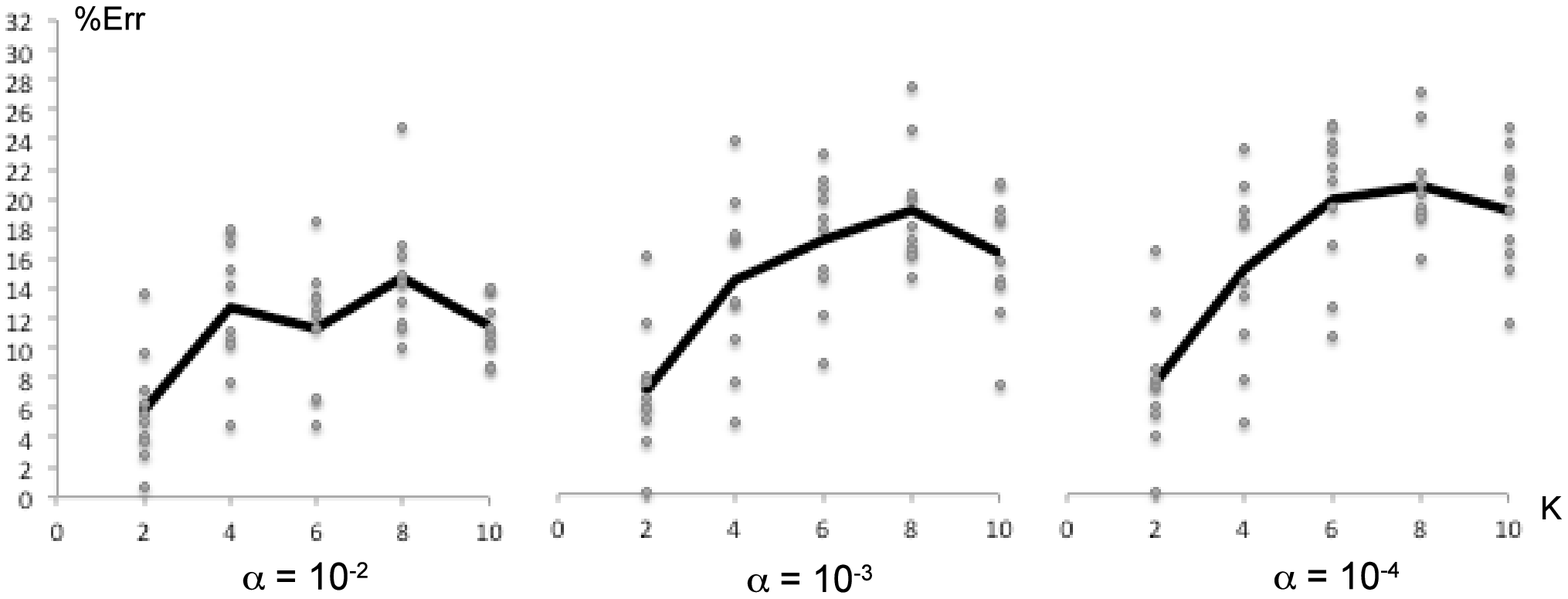}
\caption{The performance of the approximation algorithm 
 run on the instances of
  \textsc{Min-WOWA  Assignment} --
percentage deviations from the optimum for $m=40$ and all combinations of $\alpha$ and $K$. The solid line shows the average deviation.}
\label{fig6ex}       
\end{figure}
\begin{figure}[ht]
\centering
\includegraphics[scale=.6]{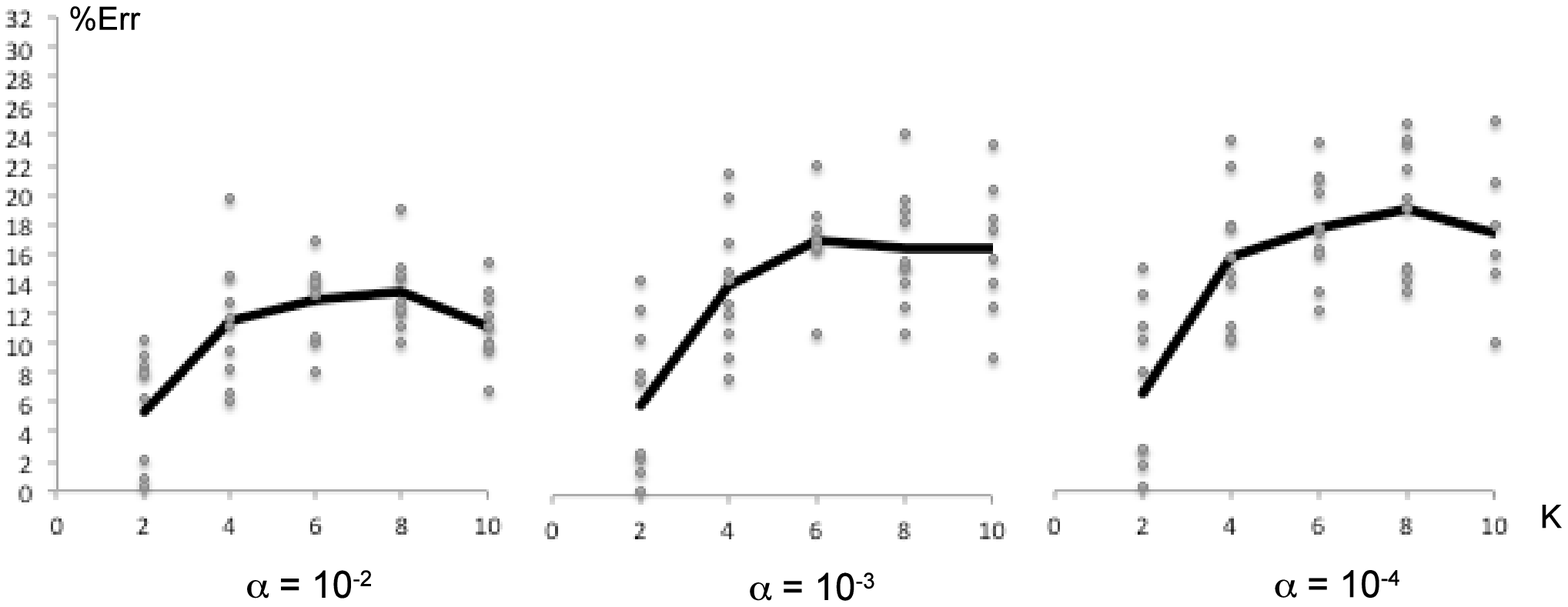}
\caption{he performance of the approximation algorithm 
 run on the instances of
  \textsc{Min-WOWA  Assignment} --
percentage deviations from the optimum for $m=50$ and all combinations of $\alpha$ and $K$. The solid line shows the average deviation.}
\label{fig7ex}       
\end{figure}

\section{Conclusions}

In this paper we have discussed a wide class of discrete optimization problems in which the uncertain costs are specified in the form of a discrete scenario set. 
A probability distribution over this set of scenarios set is provided. We have applied the weighted OWA criterion to choose a solution. This criterion allows us to take both scenario probabilities and attitude of decision makers towards a risk into account, as the weights assigned to scenarios are distorted (rank dependent) probabilities. Our approach contains the traditional robust (min-max) and stochastic approaches as special cases. The problem of minimizing the WOWA criterion is typically NP-hard for two scenarios. It becomes strongly NP-hard and also hard to approximate when the number of scenarios is  part of the input. It is thus important to provide efficient approximation algorithms for the problem. One such an algorithm has been constructed in this paper.
 It can be applied, if  the underlying deterministic problem is polynomially solvable. The efficiency of the MIP formulation and the quality of the approximation algorithm were tested for two particular problems, namely the selection and the assignment problems. The MIP formulation can be used when the number of scenarios is small. For larger number of scenarios the approximation algorithm may be an attractive choice. The performance of the approximation algorithm seems to be good when the element costs under scenarios are chosen randomly. It may be poorer for  more correlated costs and investigating the quality of the algorithm in this case requires additional tests.

\subsubsection*{Acknowledgements}
The authors would like to thank the anonymous reviewers for their valuable comments and suggestions to
 improve the quality of the paper.  This work is supported by
 the National Center for Science (Narodowe Centrum Nauki), grant  2013/09/B/ST6/01525.


\end{document}